\documentclass[envcountsame,oribibl,runningheads]{llncs}

\usepackage{latexsym}
\usepackage{amssymb}  
\usepackage{amsmath}
\usepackage{ifthen}
\usepackage{array}

\newcommand{\defemph}[1]{\textbf{\textsl{#1}}}
\newcommand{\uu}{\mathbf{u}}
\newcommand{\vv}{\mathbf{v}}

\begin{document}

\bibliographystyle{plain}

\title{On the Positivity Problem for\\ Simple Linear Recurrence
  Sequences\thanks{This research was partially supported by
    EPSRC\@. We are also grateful to Matt Daws for considerable
    assistance in the initial stages of this work.}}

\author{Jo\"el Ouaknine  
\and James Worrell}

\institute{Department of Computer Science,
Oxford University, UK}

\maketitle

\begin{abstract}
Given a linear recurrence sequence (LRS) over the integers, the
\emph{Positivity Problem} asks whether all terms of the sequence are
positive. We show that, for simple LRS (those whose characteristic
polynomial has no repeated roots) of order~$9$ or less, Positivity is
decidable, with complexity in the Counting Hierarchy.
\end{abstract}

\section{Introduction}

A (real) \defemph{linear recurrence sequence (LRS)} is an infinite
sequence $\uu = \langle u_0, u_1, u_2, \ldots \rangle$ of real numbers
having the following property: there exist constants $b_1, b_2,
\ldots, b_k$ (with $b_k \neq 0$) such that, for all $n \geq 0$,
\begin{equation}
\label{rec-rel}
u_{n+k} = b_1 u_{n+k-1} + b_2 u_{n+k-2} + \ldots + b_k u_{n} \, .
\end{equation} 
If the initial values $u_0, \ldots, u_{k-1}$ of the sequence are
provided, the recurrence relation defines the rest of the sequence
uniquely. Such a sequence is said to have \defemph{order}
$k$.\footnote{Some authors define the order of an LRS as the
  \emph{least} $k$ such that the LRS obeys such a recurrence
  relation. The definition we have chosen allows for a simpler
  presentation of our results and is algorithmically more convenient.}

The best-known example of an LRS was given by Leo\-nar\-do of Pisa in
the 12th century: the Fibonacci sequence $\langle 0, 1, 1, 2,$ $3, 5,
8, 13, \ldots \rangle$, which satisfies the recurrence relation
$u_{n+2} = u_{n+1} + u_n$. Leonardo of Pisa introduced this sequence
as a means to model the growth of an idealised population of
rabbits. Not only has the Fibonacci sequence been extensively studied
since, but LRS now form a vast subject in their own right, with
numerous applications in mathematics and other sciences. A deep and
extensive treatise on the mathematical aspects of recurrence sequences
is the monograph of Everest \emph{et al.}~\cite{BOOK}.

Given an LRS $\uu$ satisfying the recurrence relation~(\ref{rec-rel}),
the \defemph{characteristic polynomial} of $\uu$ is
\begin{equation}
\label{char-poly}
p(x) = x^k - b_1 x^{k-1} - \ldots - b_{k-1} x - b_k \, .
\end{equation}
An LRS is said to be \defemph{simple} if its characteristic polynomial
has no repeated roots. Simple LRS, such as the Fibonacci sequence,
possess a number of desirable properties which considerably simplify
their analysis---see, e.g., \cite{BOOK,ESS02,AV09,AV11,OW14a}. They
constitute a large\footnote{In the measure-theoretic sense, almost all
  LRS are \emph{simple} LRS\@.} and well-studied class of sequences,
and correspond to \emph{diagonalisable} matrices in the matricial
formulation of LRS---see Sec.~\ref{sec-LRS}.

In this paper, we focus on the \defemph{Positivity Problem} for simple
LRS over the integers (or equivalently, for our purposes, the
rationals): given a simple LRS, are all of its terms
positive?\footnote{In keeping with established terminology, `positive'
  here is taken to mean `non-negative'.}

As detailed in~\cite{OW14}, the Positivity Problem (and assorted
variants) has applications in a wide array of scientific areas,
including theoretical biology, economics, software verification,
probabilistic model checking, quantum computing, discrete linear
dynamical systems, combinatorics, formal languages, statistical
physics, generating functions, etc.
Positivity also bears an important relationship to the well-known
\emph{Skolem Problem}: does a given LRS have a zero? The decidability
of the Skolem Problem is generally considered to have been open since
the 1930s (notwithstanding the fact that algorithmic decision issues
had not at the time acquired the importance that they have
today---see~\cite{TUCS05} for a discussion on this subject; see
also~\cite[p.~258]{Tao08} and \cite{Lip09}, in which this state of
affairs---the enduring openness of decidability for the Skolem
Problem---is described as ``faintly outrageous'' by Tao and a
``mathematical embarrassment'' by Lipton). A breakthrough occurred in
the mid-1980s, when Mignotte \emph{et al.}~\cite{MST84} and
Vereshchagin~\cite{Ver85} independently showed decidability for
LRS of order $4$ or less. These deep results make essential
use of Baker's theorem on linear forms in logarithms (which earned
Baker the Fields Medal in 1970), as well as a $p$-adic analogue of
Baker's theorem due to van der Poorten. Unfortunately, little progress
on that front has since been recorded.\footnote{A proof of
  decidability of the Skolem Problem for LRS of order 5 was announced
  in~\cite{TUCS05}. However, as pointed out in~\cite{OW12}, the proof
  seems to have a serious gap.}

It is considered folklore that the decidability of Positivity (for
arbitrary LRS) would entail that of the Skolem Problem~\cite{OW14},
noting however that the reduction increases the order of LRS
quadratically.\footnote{It is worth noting that, under this reduction,
  the decidability of the Positivity Problem for simple LRS of order
  at most $14$ would entail the decidability of the Skolem Problem for
  simple LRS of order~$5$, which is open and from which the general
  case of the Skolem Problem at order~$5$ would follow, based on the
  work carried out in~\cite{TUCS05}; see also~\cite{OW12}, which
  identifies the last unresolved critical case for the Skolem Problem
  at order~$5$, involving simple LRS\@.}  Nevertheless, the earliest
explicit references in the literature to the Positivity Problem that
we have found are from the 1970s (see,
e.g.,~\cite{Soi76,Sal76,BM76}). In~\cite{Soi76}, the Skolem and
Positivity Problems are described as ``very difficult'', whereas
in~\cite{RS94}, the authors assert that the Skolem and Positivity
Problems are ``generally conjectured [to be] decidable''. Positivity
is again stated as an open problem
in~\cite{HHH06,BG07,LT09,Liu10,TV11,OW14}, among others.

Unsurprisingly, progress on the Positivity Problem over the last few
decades has been fairly slow. In the early 1980s, Burke and
Webb~\cite{BW81} showed that the closely related problem of
\emph{Ultimate Positivity} (are all but finitely many terms of a given
LRS positive?)\ is decidable for LRS of order~$2$, and nine years
later Nagasaka and Shiue~\cite{NS90} showed the same for LRS of
order~$3$ that have repeated characteristic roots. Much more recently,
Halava \emph{et al.}~\cite{HHH06} showed that Positivity is decidable
for LRS of order~$2$, and subsequently Laohakosol and
Tangsupphathawat~\cite{LT09} proved that Positivity is decidable for
LRS of order~$3$.  In 2012, an article purporting to show decidability
of Positivity for LRS of order~$4$ was published~\cite{TPL12}, with
the authors noting they were unable to tackle the case of
order~$5$. Unfortunately, as pointed out in~\cite{OW14} and
acknowledged by the authors themselves~\cite{Lao13}, that paper
contains a major error, invalidating the order-$4$ claim. Very
recently, Positivity was nevertheless shown decidable for arbitrary
integer LRS of order~$5$ or less~\cite{OW14}, with complexity in the
Counting Hierarchy; moreover, the same paper shows by way of hardness
that the decidability of Positivity for integer LRS of order~$6$ would
entail major breakthroughs in analytic number theory (certain
longstanding Diophantine-approximation open problems would become
solvable). Finally, in~\cite{OW14a}, the authors show that Ultimate Positivity
for simple integer LRS of unrestricted order is decidable within
PSPACE, and in polynomial time if the order is fixed.

\

\noindent
\textbf{Main Result.}
The main result of this paper is that the Positivity Problem for
simple integer LRS of order~$9$ or less is decidable. An analysis of
the decision procedure shows that its complexity lies in
$\mathrm{coNP}^{\mathrm{PP}^{\mathrm{PP}^{\mathrm{PP}}}}$, i.e.,
within the fourth level of the Counting Hierarchy (itself contained in
PSPACE)\@.\footnote{The complexity is as a function of the bit length
  of the standard representation of integer LRS; for an LRS of order
  $k$ as defined by Eq.~(\ref{rec-rel}), this representation
  consists of the $2k$-tuple $(b_1, \ldots, b_k, u_0, \ldots,
  u_{k-1})$ of integers.

%Note also that the complexity class does not
%  require parenthesising since
%  $\mathrm{co(NP}^{\mathrm{PP}^{\mathrm{PP}^{\mathrm{PP}}}}\mathrm{)}
%  = \mathrm{(coNP)}^{\mathrm{PP}^{\mathrm{PP}^{\mathrm{PP}}}}$.
} 

\

\noindent
\textbf{Comparison with Related Work.}
It is important to note the fundamental difference between the above
result and those of~\cite{OW14a}: in the latter, Ultimate Positivity
is shown to be decidable for simple LRS of all orders, but in a
\emph{non-constructive} sense: a given LRS may be certified ultimately
positive, yet no index threshold is provided beyond which all terms of
the LRS are positive. At the time of writing, this appears to be a
fundamental difficulty: for simple LRS of any given order, the ability
to compute such index thresholds would immediately enable one to
decide Positivity. Yet as noted earlier, the decidability of
Positivity for simple LRS of order at most $14$ would in turn entail
the decidability of the Skolem Problem for arbitrary LRS of order $5$,
a longstanding and major open problem.

We recall and summarise some standard material on linear recurrence
sequences in Sec.~\ref{sec-LRS}. We also recall in App.~\ref{tools}
the statements of various mathematical tools needed in our
development, notably Baker's theorem on linear forms in logarithms,
Masser's results on multiplicative relationships among algebraic
numbers, Kronecker's theorem on simultaneous Diophantine
approximation, and Renegar's work on the fine-grained complexity of
quantifier elimination in the first-order theory of the reals.

Our overall approach is similar to that followed in~\cite{OW14},
attacking the problem via the exponential polynomial solution of LRS
using sophisticated tools from analytic and algebraic number theory,
Diophantine geometry and approximation, and real algebraic
geometry. However the present paper makes vastly greater and deeper
use of real algebraic geometry, particularly in the form of
Lemmas~\ref{lem-one-constraint}, \ref{lem-m-1-constraints}, and
\ref{zero-dim-lemma} (which serve to establish the key fact that
certain varieties are zero-dimensional, enabling our application of
Baker's theorem in higher dimensions), and throughout the whole of
Sec.~\ref{7-dom-roots}, which handles what is by far the most
difficult and complex critical case in our analysis. 
Lemmas~\ref{lem-one-constraint}--\ref{zero-dim-lemma} (which
can be found in App.~\ref{app-lemmas}), as well as the development
of Sec.~\ref{7-dom-roots}, are entirely new.

The present paper also markedly differs from~\cite{OW14a}. In fact,
aside from sharing standard material on LRS, the non-constructive
approach of~\cite{OW14a} eschews most of the real algebraic geometry
of the present paper, as well as Baker's theorem, and is
underpinned instead by non-constructive lower bounds on sums of
$S$-units, which in turn follow from deep results in Diophantine
approximation (Schlickewei's $p$-adic generalisation of Schmidt's
Subspace theorem).

\

\noindent
We present a high-level over\-view of our proof strategy---split in
two parts---within Sec.~\ref{decidability}, and also briefly
discuss why the present approach does not seem extendable beyond
order~$9$. As noted earlier, establishing the decidability of
Positivity for simple LRS of order~$14$ would entail a major advance,
namely the decidability of the Skolem Problem for arbitrary LRS of
order $5$. It is an open problem whether similar `hardness' results
can be established for simple LRS of orders~$10$--$13$.

In terms of complexity, it is shown in~\cite{OW14a} that the
Positivity Problem for simple integer LRS of arbitrary order is hard
for co$\exists \mathbb{R}$, the class of problems whose complements
are solvable in the existential theory of the reals, and which is
known to contain coNP\@. However, no lower bounds are known when the
order is fixed or bounded, as is the case in the present paper. Either
establishing non-trivial lower bounds or improving the
Counting-Hierarchy complexity of the present procedure also appear to
be challenging open problems.

\section{Linear Recurrence Sequences}
\label{sec-LRS}

We recall some fundamental properties of (simple) linear recurrence
sequences.  Results are stated without proof, and we refer the reader
to~\cite{BOOK,TUCS05} for details.

Let $\uu = \langle u_n \rangle_{n=0}^{\infty}$ be an LRS of order $k$
over the reals satisfying the recurrence relation $u_{n+k} = b_1
u_{n+k-1} + \ldots + b_k u_{n}$, where $b_k \neq 0$. We denote by
$||\uu||$ the bit length of its representation as a $2k$-tuple of
integers, as discussed in the previous section.  The
\defemph{characteristic roots} of $\uu$ are the roots of its
characteristic polynomial (cf.~Eq.~(\ref{char-poly})), and the
\defemph{dominant roots} are the roots of maximum modulus. The
characteristic roots can be computed in time polynomial in
$||\uu||$---see App.~\ref{tools} for further details on
algebraic-number manipulations.

The characteristic roots divide naturally into real and non-real ones.
% those that are real and those that are not. 
Since the characteristic polynomial has real coefficients, non-real
roots always arise in conjugate pairs. Thus we may write $\{\rho_1,
\ldots, \rho_{\ell}, \gamma_1, \overline{\gamma_1}, \ldots, \gamma_m,
\overline{\gamma_m}\}$ to represent the set of characteristic roots of
$\uu$, where each $\rho_i \in \mathbb{R}$ and each $\gamma_j \in
\mathbb{C} \setminus \mathbb{R}$. If $\uu$ is a simple LRS, there are
algebraic constants $a_1, \ldots, a_{\ell} \in \mathbb{R}$ and $c_1,
\ldots, c_m$ such that, for all $n \geq 0$,
\begin{equation}
\label{star}
u_n = \sum_{i=1}^{\ell} a_i \rho_i^n + 
      \sum_{j=1}^m \left(c_j \gamma_j^n + 
                  \overline{c_j} \overline{\gamma_j}^n\right) \,.
\end{equation}

This expression is referred to as the \defemph{exponential polynomial}
solution of $\uu$. For fixed $k$, all constants $a_i$ and $c_j$ can be
computed in time polynomial in $||\uu||$, since they can be obtained
by solving a system of linear equations involving the first $k$ instances
of Eq.~(\ref{star}). 

An LRS is said to be \defemph{non-degenerate} if it does not have two
distinct characteristic roots whose quotient is a root of unity.  As
pointed out in~\cite{BOOK}, the study of arbitrary LRS can effectively
be reduced to that of non-degenerate LRS, by partitioning the original
LRS into finitely many subsequences, each of which is
non-degenerate. In general, such a reduction will require exponential
time. However, when restricting to LRS of bounded order (in our case,
of order at most $9$), the reduction can be carried out in polynomial
time. In particular, any LRS of order~$9$ or less can be partitioned
in polynomial time into at most $3.9 \cdot 10^7$ non-degenerate LRS of
the same order or less.\footnote{We obtained this value using a
  bespoke enumeration procedure for order~$9$. A bound of $e^{2
    \sqrt{6\cdot 9 \log 9}} \leq 2.9 \cdot 10^9$ can be obtained from
  Cor.~3.3 of~\cite{YLN95}.} Note that if the original LRS is
simple, this process will yield a collection of simple non-degenerate
subsequences. In the rest of this paper, we shall therefore assume
that all LRS we are given are non-degenerate.

Any LRS $\uu$ of order $k$ can alternately be given in matrix form, in
the sense that there is a square matrix $M$ of dimension $k \times k$,
together with $k$-dimensional column vectors $\vec{v}$ and $\vec{w}$,
such that, for all $n \geq 0$, $u_n = \vec{v}^T M^n \vec{w}$. It
suffices to take $M$ to be the transpose of the companion matrix of
the characteristic polynomial of $\uu$, let $\vec{v}$ be the vector
$(u_{k-1},\ldots,u_0)$ of initial terms of $\uu$ in reverse order, and
take $\vec{w}$ to be the vector whose first $k-1$ entries are $0$ and
whose $k$th entry is $1$. It is worth noting that the characteristic
roots of $\uu$ correspond precisely to the eigenvalues of $M$, and
that if $\uu$ is simple then $M$ is diagonalisable. This translation
is instrumental in Sec.~\ref{decidability} to place the Positivity
Problem for simple LRS of order at most $9$ within the Counting
Hierarchy.

Conversely, given any square matrix $M$ of dimension $k \times k$, and
any $k$-dimensional vectors $\vec{v}$ and $\vec{w}$, let $u_n =
\vec{v}^T M^n \vec{w}$. Then $\langle \vec{v}^T M^n \vec{w}
\rangle_{n=k}^{\infty}$ is an LRS of order at most $k$ whose
characteristic polynomial divides that of $M$, as can be seen by
applying the Cayley-Hamilton Theorem.\footnote{In fact, if none of the
  eigenvalues of $M$ are zero, it is easy to see that the full
  sequence $\langle \vec{v}^T M^n \vec{w} \rangle_{n=0}^{\infty}$ is
  an LRS (of order at most $k$).} When $M$ is diagonalisable, the
resulting LRS is simple.

\section{Decidability and Complexity}
\label{decidability}

Let $\uu = \langle u_n \rangle_{n=0}^{\infty}$ be an integer LRS of
order $k$. As discussed in the Introduction, we assume that $u$ is
presented as a $2k$-tuple of integers $(b_1, \ldots, b_k, u_0, \ldots,
u_{k-1}) \in \mathbb{Z}^{2k}$, such that for all $n \geq 0$,
\begin{equation}
\label{recurrence}
u_{n+k} = b_1 u_{n+k-1} + \ldots + b_k u_n \, .
\end{equation}

The \defemph{Positivity Problem} asks, given such an LRS $\uu$,
whether for all $n \geq 0$, it is the case that $u_n \geq 0$. When
this holds, we say that $\uu$ is \defemph{positive}.

%The \defemph{(effective) Ultimate Positivity Problem} asks, given such
%an LRS $\uu$, whether there exists a threshold $N \geq 0$ such that,
%for all $n \geq N$, it is the case that $u_n \geq 0$. When this holds,
%we say that $\uu$ is \defemph{ultimately
%  positive}. \emph{Effectiveness} requires in addition that the
%threshold $N$ be explicitly produced.

In this section, we establish the following:

\begin{theorem}
\label{pos-theorem}
The Positivity Problem for simple integer LRS of order $9$ or less is
decidable in
$\mathrm{coNP}^{\mathrm{PP}^{\mathrm{PP}^{\mathrm{PP}}}}$. 
\end{theorem}

Note that deciding whether the characteristic roots are simple can
easily be done in polynomial time; cf.\ App.~\ref{tools}.

Observe also that Thm.~\ref{pos-theorem} immediately carries over to
rational LRS\@. To see this, consider a rational LRS $\uu$ obeying the
recurrence relation~(\ref{recurrence}). Let $\ell$ be the least common
multiple of the denominators of the rational numbers $b_1, \ldots,
b_k, u_0, \ldots, u_{k-1}$, and define an integer sequence $\vv =
\langle v_n \rangle_{n=0}^{\infty}$ by setting $v_n = \ell^{n+1} u_n$
for all $n \geq 0$. It is easily seen that $\vv$ is an integer LRS of
the same order as $\uu$, and that for all $n$, $v_n \geq 0$ iff $u_n
\geq 0$. Moreover, $\vv$ is simple iff $\uu$ is simple.

\

\noindent
\textbf{High-Level Synopsis (I)\@.}
At a high level, the algorithm upon which Thm.~\ref{pos-theorem} rests
proceeds as follows. Given an LRS $\uu$, we first decide whether or
not $\uu$ is ultimately positive\footnote{A sequence is ultimately
  positive if all but finitely many of its terms are positive.} by
studying its exponential polynomial solution---further details on this
task are provided shortly. As we prove in this paper, whenever $\uu$
is an ultimately positive simple LRS of order~$9$ or less, there is a
polynomial-time computable threshold $N$ of at most exponential
magnitude such that all terms of $\uu$ beyond $N$ are
positive. Clearly $\uu$ cannot be positive unless it is ultimately
positive. Now in order to assert that an ultimately positive LRS $\uu$
is \emph{not} positive, we use a \emph{guess-and-check} procedure:
find $n \leq N$ such that $u_n < 0$. By writing $u_n = \vec{v}^T M^n
\vec{w}$, for some square integer matrix $M$ and vectors $\vec{v}$ and
$\vec{w}$ (cf.~Sec.~\ref{sec-LRS}), we can decide whether $u_n < 0$ in
$\mathrm{PosSLP}$\footnote{Recall that $\mathrm{PosSLP}$ is the
  problem of determining whether an arithmetic circuit, with addition,
  multiplication, and subtraction gates, evaluates to a positive
  integer.} via iterative squaring, which yields an
$\mathrm{NP}^{\mathrm{PosSLP}}$ procedure for non-Positivity. Thanks
to the work of Allender \emph{et al.}~\cite{ABK09}, which asserts that
$\mathrm{PosSLP} \subseteq
\mathrm{P}^{\mathrm{PP}^{\mathrm{PP}^{\mathrm{PP}}}}$, we obtain the
required $\mathrm{coNP}^{\mathrm{PP}^{\mathrm{PP}^{\mathrm{PP}}}}$
algorithm for deciding Positivity.

\

\noindent
The following is an old result concerning LRS; proofs can be found
in~\cite[Thm.~7.1.1]{GL91} and \cite[Thm.~2]{BG07}. It also follows
easily and directly from either Pringsheim's theorem or
from~\cite[Lem.~4]{Bra06}. It plays an important role in our approach
by enabling us to significantly cut down on the number of subcases
that must be considered, avoiding the sort of quagmire alluded to
in~\cite{NS90}.

\begin{proposition}
\label{prune}
Let $\langle u_n \rangle_{n=0}^{\infty}$ be an LRS with no real
positive dominant characteristic root. Then there are infinitely many
$n$ such that $u_n < 0$ and infinitely many $n$ such that $u_n > 0$.
\end{proposition}

By Prop.~\ref{prune}, it suffices to restrict our attention to
LRS whose dominant characteristic roots include one real positive
value. Given an integer LRS $\uu$, note that determining whether the
latter holds is easily done in time polynomial in $||\uu||$
(cf.\ App.~\ref{tools}).

Thus let $\uu$ be a non-degenerate simple integer LRS of order $k \leq
9$ having a real positive dominant characteristic root $\rho >
0$. Note that $\uu$ cannot have a real negative dominant
characteristic root (which would be $-\rho$), since otherwise the
quotient $-\rho/\rho = -1$ would be a root of unity, contradicting
non-degeneracy.  Let us 
%therefore 
write the characteristic roots as $\{\rho, \gamma_1,
\overline{\gamma_1}, \ldots, \gamma_m, \overline{\gamma_m}\} \cup
\{\gamma_{m+1}, \gamma_{m+2}, \ldots, \gamma_{\ell}\}$, where we
assume that the roots in the first set all have common modulus $\rho$,
whereas the roots in the second set all have modulus strictly smaller
than $\rho$.

Let $\lambda_i = \gamma_i/\rho$ for $1 \leq i \leq \ell$. We can
then write
\begin{equation}
\label{LRSrep1}
\frac{u_n}{\rho^n} = a + 
          \sum_{j=1}^m \left(c_j \lambda_j^n + 
          \overline{c_j} \overline{\lambda_j}^n\right)+ r(n) \, ,
\end{equation}
for some real algebraic constant $a$ and complex algebraic
constants $c_1, \ldots, c_m$, where $r(n)$ is a term tending to zero
exponentially fast.

Note that none of $\lambda_1, \ldots, \lambda_m$, all of which have
modulus $1$, can be a root of unity, as each $\lambda_i$ is a quotient
of characteristic roots and $\uu$ is assumed to be
non-degenerate. Likewise, for $i \neq j$, $\lambda_i/\lambda_j$ and 
$\overline{\lambda_i}/\lambda_j$ cannot be roots of unity. 

For $i \in \{1, \ldots, \ell\}$, observe also that as each $\lambda_i$
is a quotient of two roots of the same polynomial of degree $k$, it
has degree at most $k(k-1)$. In fact, it is easily seen that
$||\lambda_i|| = ||\uu||^{\mathcal{O}(1)}$, $||a|| =
||\uu||^{\mathcal{O}(1)}$, and $||c_i|| = ||\uu||^{\mathcal{O}(1)}$
(cf.\ App.~\ref{tools}).

Finally, we place bounds on the rate of convergence of $r(n)$.  We
have 
\[ 
r(n) = c_{m+1} \lambda_{m+1}^n + \ldots +
c_{\ell}\lambda_{\ell}^n \, .
\] 
Combining our estimates on the height and
degree of each $\lambda_i$ together with the root-separation bound
given by Eq.~(\ref{root-sep-bound}) in App.~\ref{tools}, we
get $\frac{1}{1-|\lambda_i|} =
2^{||\uu||^{\mathcal{O}(1)}}$, for $m+1 \leq i \leq \ell$. Thanks also
to the bounds on the height and degree of the constants $c_i$, it
follows that we can find $\varepsilon \in (0,1)$ and $N \in
\mathbb{N}$ such that:
\begin{align}
\label{eq1}
& 1/\varepsilon = 2^{||\uu||^{\mathcal{O}(1)}} \\
& N = 2^{||\uu||^{\mathcal{O}(1)}}  \\
& \mbox{For all } n > N,\ |r(n)| < (1-\varepsilon)^n \, .
\label{eq3}
\end{align}
We can compute such $\varepsilon$ and $N$ in time polynomial in
$||\uu||$, since all relevant calculations on algebraic numbers
only require polynomial time (cf.\ App.~\ref{tools}).
%Naturally, given $k$, we can also assume that
%we have calculated explicitly once and for all the constants implicit
%in the various instances of the $\mathcal{O}(1)$ notation.

We now seek to answer positivity and ultimate positivity questions for
the LRS $\uu = \langle u_n \rangle_{n=0}^{\infty}$ by studying the
same for $\langle u_n/\rho^n \rangle_{n=0}^{\infty}$.

In what follows, we assume that $\uu$ is as above, i.e.,
$\uu$ is a non-degenerate simple integer LRS having a real
positive dominant characteristic root $\rho > 0$.

\

\noindent
\textbf{High-Level Synopsis (II)\@.}  Before launching into technical
details, let us provide a high-level overview of our proof strategy
for deciding whether $\uu$ is ultimately positive, and when that is
the case, for computing an index threshold $N$ beyond which all of its
terms are positive.  Let us rewrite Eq.~(\ref{LRSrep1}) as
\begin{equation}
\label{eq-explanation}
\frac{u_n}{\rho^n} = a + h(\lambda_1^n, \ldots, \lambda_m^n) 
                              + r(n) \, ,
\end{equation}
where $h:\mathbb{C}^m \rightarrow \mathbb{R}$ is a continuous
function. In general, there will be integer multiplicative
relationships among the $\lambda_1, \ldots, \lambda_m$, forming a free
abelian group $L$ for which we can compute a basis thanks to
Thm.~\ref{Ge} in App.~\ref{tools}. These multiplicative
relationships define a torus $T \subseteq \mathbb{C}^m$ on which the
joint iterates $\{(\lambda_1^n, \ldots, \lambda_m^n) : n \in
\mathbb{N} \}$ are dense, as per Kronecker's theorem (in the form of
Cor.~\ref{density} in App.~\ref{tools}).

Now the critical case arises when
%\[ a + \min h\mbox{$\restriction$}_T = 0 \, , \]
$a + \min h\mbox{$\restriction$}_T = 0$,
where $h\mbox{$\restriction$}_T$ denotes the function $h$ restricted
to the torus $T$. Provided that $h\mbox{$\restriction$}_T$ achieves
its minimum $-a$ at only finitely many points, we can use Baker's
theorem (in the form of Cor.~\ref{Baker-cor} in
App.~\ref{tools}) to bound the iterates $(\lambda_1^n, \ldots,
\lambda_m^n)$ away from these points by an inverse polynomial in
$n$. By combining Renegar's results (Thm.~\ref{renegar} in
App.~\ref{tools}) with techniques from real algebraic geometry, we
then argue that $h(\lambda_1^n, \ldots, \lambda_m^n)$ is itself
eventually bounded away from the minimum $-a$ by a (different) inverse
polynomial in $n$, and since $r(n)$ decays to zero exponentially fast,
we are able to conclude that $u_n/\rho^n$ is ultimately positive, and
can compute a threshold $N$ after which all terms $u_n$ (for $n > N$)
are positive.

Note in the above that a key component is the requirement that
$h\mbox{$\restriction$}_T$ achieve its minimum at finitely many
points. Lemmas~\ref{lem-one-constraint}--\ref{zero-dim-lemma}
in App.~\ref{app-lemmas} show that this is the case provided that
$L$, the free abelian group of multiplicative relationships among the
$\lambda_1, \ldots, \lambda_m$, has rank $0$, $1$, $m-1$, or
$m$. In fact, simple counterexamples can be manufactured in the other
instances, which seems to preclude the use of Baker's theorem. Since
non-real characteristic roots always arise in conjugate pairs, the
earliest appearance of this vexing state of affairs is at order $10$:
one real dominant root, $m=4$ pairs of complex dominant roots, one
non-dominant root ensuring that the term $r(n)$ is not identically
$0$, and a free abelian group $L$ of rank $2$. The difficulty encountered
there is highly reminiscent of (if technically different from) that of
the critical unresolved case for the Skolem Problem at order~$5$, as
described in~\cite{OW12}.

\

\noindent
We now proceed with the formalisation of the above. Recall that $\uu$
is assumed to be a non-degenerate simple LRS of order at most $9$,
with a real positive dominant characteristic root $\rho > 0$ and
complex dominant roots $\gamma_1, \overline{\gamma_1}, \ldots,
\gamma_m, \overline{\gamma_m} \in \mathbb{C} \setminus \mathbb{R}$. We
write $\lambda_j = \gamma_j/\rho$ for $1 \leq j \leq m$.

Since the number of dominant roots is odd and at most $9$, we divide
our analysis into two cases, there being exactly $9$ dominant roots
(Sec.~\ref{9-dom-roots}), and there being $7$ or fewer dominant
roots (Sec.~\ref{7-dom-roots}). Our starting point is
Eq.~(\ref{LRSrep1}).

Let $L = \{(v_1, \ldots, v_m) \in \mathbb{Z}^m : \lambda_1^{v_1}
\ldots \lambda_m^{v_m} = 1\}$ have rank $p$ (as a free abelian group),
and let $\{\vec{\ell_1}, \ldots, \vec{\ell_p}\}$ be a basis for $L$.
Write $\vec{\ell_q} = (\ell_{q,1}, \ldots, \ell_{q,m})$ for $1 \leq q
\leq p$. Recall from Thm.~\ref{Ge} in App.~\ref{tools} that
such a basis may be computed in polynomial time, and moreover that
each $\ell_{q,j}$ may be assumed to have magnitude polynomial in
$||\uu||$.

\subsection{Nine Dominant Roots}
\label{9-dom-roots}
If $\uu$ has 9 dominant roots, then $m =4$ and $r(n)$ is identically
$0$ in Eq.~(\ref{LRSrep1}); the latter greatly simplifies our
analysis.

Write $\mathbb{T} = \{z \in \mathbb{C} : |z| = 1\}$ for the unit
circle in the complex plane, and let
\[
T = \{(z_1, \ldots, z_4) \in \mathbb{T}^4
: \mbox{for each $q \in \{1, \ldots, p\}$, } 
z_1^{\ell_{q,1}} \ldots z_4^{\ell_{q,4}} = 1 \} \, .
\]

Define $h : T \rightarrow \mathbb{R}$ by $h(z_1, \ldots, z_4)
= \sum_{j=1}^4 (c_j z_j + \overline{c_j} \overline{z_j})$, so that for
all $n$, 
\[
\frac{u_n}{\rho^n} = a + h(\lambda_1^n, \ldots, \lambda_4^n) \, .
\] 
By Cor.~\ref{density} in App.~\ref{tools}, the set
$\{(\lambda_1^n, \ldots, \lambda_4^n) : n \in \mathbb{N}\}$ is a dense
subset of $T$. Since $h$ is continuous, we then have that $\inf
\{u_n/\rho^n : n \in \mathbb{N}\} = a + \min
h\mbox{$\restriction$}_T$. It follows that $\uu$ is ultimately
positive iff $\uu$ is positive iff $\min h\mbox{$\restriction$}_T \geq
-a$ iff
\begin{equation}
\label{min-eqn}
 \forall (z_1, z_2, z_3, z_4) \in T, \, h(z_1, z_2, z_3, z_4) \geq -a \, .
\end{equation}

We now show how to rewrite Assertion~(\ref{min-eqn}) as a sentence in
the first-order theory of the reals, i.e., involving only real-valued
variables and first-order quantifiers, Boolean connectives, and
integer constants together with the arithmetic operations of addition,
subtraction, multiplication, and division.\footnote{In
  App.~\ref{tools}, we do not have division as an allowable
  operation in the first-order theory of the reals; however instances
  of division can always be removed in linear time at the cost of
  introducing a linear number of existentially quantified fresh
  variables.} The idea is to separately represent the real and
imaginary parts of each complex quantity appearing in
Assertion~(\ref{min-eqn}), and combine them using real arithmetic so
as to mimic the effect of complex arithmetic operations.

To this end, we use two real variables $x_j$ and $y_j$ to represent
each of the $z_j$: intuitively, $z_j = x_j + i y_j$. Since the real
constant $a$ is algebraic, there is a formula $\sigma_a(x)$ which is
true over the reals precisely for $x = a$. Likewise, the real and
imaginary parts $\mathrm{Re}(c_j)$ and $\mathrm{Im}(c_j)$ of the
complex algebraic constants $c_j$ are themselves real algebraic, and
can be represented as formulas in the first-order theory of the
reals. All such formulas can readily be shown to have size polynomial
in $||u||$.

Terms of the form $z_j^{\ell_{q,j}}$ are simply expanded: for example,
if $\ell_{q,j}$ is positive, then $z_j^{\ell_{q,j}} = (x_j + i
y_j)^{\ell_{q,j}} = A_{q,j}(x_j) + i B_{q,j}(y_j)$, where $A_{q,j}$
and $B_{q,j}$ are polynomials with integer coefficients. Note that
since the magnitude of $\ell_{q,j}$ is polynomial in $||\uu||$, so are
$||A_{q,j}||$ and $||B_{q,j}||$. The case in which $\ell_{q,j}$ is
negative is handled similarly, with the additional use of a division
operation.

Combining everything, we obtain a sentence $\tau$ of the first-order
theory of the reals with division which is true iff
Assertion~(\ref{min-eqn}) holds. $\tau$ makes use of at most $17$ real
variables: two for each of $z_1, \ldots, z_4$, one for $a$, and one
for each of $\mathrm{Re}(c_1), \mathrm{Im}(c_1),$ $\ldots,
\mathrm{Re}(c_4), \mathrm{Im}(c_4)$. In removing divisions from
$\tau$, the number of variables potentially swells to $29$. Finally,
the size of $\tau$ is polynomial in $||\uu||$. We can therefore invoke
Thm.~\ref{renegar} in App.~\ref{tools} to conclude that
Assertion~(\ref{min-eqn})---and therefore the positivity of
$\uu$---can be decided in time polynomial in $||\uu||$.

\subsection{Seven or Fewer Dominant Roots} 
\label{7-dom-roots}
We now turn to the main case, i.e., the situation in which $\uu$ has
$7$ dominant roots, so that $m = 3$ in Eq.~(\ref{LRSrep1}). The
cases of $1$, $3$, and $5$ dominant roots are very similar---if
slightly simpler---and are therefore omitted.

As before, we let $\mathbb{T} = \{ z \in \mathbb{C} : |z|=1\}$, and 
write
\[
T = \{(z_1, z_2, z_3) \in \mathbb{T}^3
: \mbox{for each $q \in \{1, \ldots, p\}$, }
z_1^{\ell_{q,1}} z_2^{\ell_{q,2}} z_3^{\ell_{q,3}} = 1 \} \, .
\]

Define $h : T \rightarrow \mathbb{R}$ by $h(z_1, z_2, z_3) =
\sum_{j=1}^3 (c_j z_j + \overline{c_j} \overline{z_j})$, so that for
all $n$, 
\begin{equation}
\label{eq-seven}
\frac{u^n}{\rho^n} = a + h(\lambda_1^n, \lambda_2^n, \lambda_3^n) +
r(n) \, .
\end{equation} 
By Cor.~\ref{density} in App.~\ref{tools}, the set
$\{(\lambda_1^n, \lambda_2^n, \lambda_3^n) : n \in \mathbb{N}\}$ is a
dense subset of $T$. Since $h$ is continuous, we have $\inf
\{h(\lambda_1^n, \lambda_2^n, \lambda_3^n) : n \in \mathbb{N}\} = \min
h\mbox{$\restriction$}_T = \mu$, for some $\mu \in \mathbb{R}$.

We can represent $\mu$ via the following formula $\tau(y)$:
\[
\exists(\zeta_1,\zeta_2,\zeta_3) \in T : (h(\zeta_1,\zeta_2,\zeta_3)
= y \wedge
\forall(z_1,z_2,z_3) \in T,\, y \leq h(z_1,z_2,z_3)) \, .
\]
Similarly to the translation carried out in Sec.~\ref{9-dom-roots},
we can construct an equivalent formula $\tau'(y)$ in the first-order
theory of the reals, over a bounded number of real variables, with
$||\tau'(y)|| = ||\uu||^{\mathcal{O}(1)}$. According to
Thm.~\ref{renegar} in App.~\ref{tools}, we can then compute in
polynomial time an equivalent quantifier-free formula
\[ \chi(y) = \bigvee_{i=1}^I 
\bigwedge_{j=1}^{J_i} h_{i,j}(y) \sim_{i,j} 0 \, .
\]

Recall that each $\sim_{i,j}$ is either $>$ or $=$. Now $\chi(y)$ must
have a satisfiable disjunct, and since the satisfying assignment to
$y$ is unique (namely $y = \mu$), this disjunct must comprise at least
one equality predicate. Since Thm.~\ref{renegar} guarantees that
the degree and height of each $h_{i,j}$ are bounded by
$||\uu||^{\mathcal{O}(1)}$ and $2^{||\uu||^{\mathcal{O}(1)}}$
respectively, we immediately conclude that $\mu$ is an algebraic
number and moreover that $||\mu|| = ||\uu||^{\mathcal{O}(1)}$.

Returning to Eq.~(\ref{eq-seven}), we see that if $a + \mu < 0$,
then $\uu$ is neither positive nor ultimately positive, whereas if $a
+ \mu > 0$ then $\uu$ is ultimately positive. In the latter case,
thanks to our bounds on $||\mu||$, together with the root-separation
bound given by Eq.~(\ref{root-sep-bound}) in
App.~\ref{tools}, we have $\frac{1}{a + \mu} =
2^{||\uu||^{\mathcal{O}(1)}}$. The latter, together with
Eqs.~(\ref{eq1})--(\ref{eq3}), implies an exponential upper bound
on the index of possible violations of positivity.  The actual
positivity of $\uu$ can then be decided via a coNP procedure that
invokes a PosSLP oracle as outlined earlier.

It remains to analyse the case in which $\mu = - a$. To this end, let
$\lambda_j = e^{i \theta_j}$ for $1 \leq j \leq 3$. From
Eq.~(\ref{LRSrep1}), we have:
%\begin{equation}
\[
%\label{trig-eq}
\frac{u_n}{\rho^n} =
a + \sum_{j=1}^3 2|c_j| \cos(n \theta_j + \varphi_j) + r(n) \, .
\]
%\end{equation}
In the above, $c_j = |c_j| e^{i \varphi_j}$ for $1 \leq j \leq 3$. We
make the further assumption that each $c_j$ is non-zero; note that if
this did not hold, we could simply recast our analysis in a lower
dimension.

Let $Z = \{(\zeta_1,\zeta_2,\zeta_3) \in T :
h(\zeta_1,\zeta_2,\zeta_3) = \mu \}$ be the set of points of $T$ at
which $h$ achieves its minimum $\mu$. By Lem.~\ref{zero-dim-lemma} in
App.~\ref{app-lemmas}, $Z$ is finite. We concentrate on the set
$Z_1$ of first coordinates of $Z$. Write
\begin{align*}
\tau_1(x) = & \ \exists z_1 : 
(\mathrm{Re}(z_1) = x \wedge z_1 \in Z_1) \\
\tau_2(y) = & \ \exists z_1 : 
(\mathrm{Im}(z_1) = y \wedge z_1 \in Z_1) \, .
\end{align*}
Similarly to our earlier constructions, $\tau_1(x)$ is equivalent to a
formula $\tau'_1(x)$ in the first-order theory of the reals, over a
bounded number of real variables, with $||\tau'_1(x)|| =
||\uu||^{\mathcal{O}(1)}$. Thanks to Thm.~\ref{renegar} in
App.~\ref{tools}, we then obtain an equivalent quantifier-free
formula
\[ \chi_1(x) = \bigvee_{i=1}^I 
\bigwedge_{j=1}^{J_i} h_{i,j}(x) \sim_{i,j} 0 \, .
\]

Note that since there can only be finitely many $\hat{x} \in
\mathbb{R}$ such that $\chi_1(\hat{x})$ holds, each disjunct of
$\chi_1(x)$ must comprise at least one equality predicate, or can
otherwise be entirely discarded as having no solution. 

A similar exercise can be carried out with $\tau_2(y)$, yielding
$\chi_2(y)$. The bounds on the degree and height of each $h_{i,j}$ in
$\chi_1(x)$ and $\chi_2(y)$ then enable us to conclude that any
$\zeta = \hat{x} + i \hat{y} \in Z_1$ is algebraic, and moreover
satisfies $||\zeta|| = ||\uu||^{\mathcal{O}(1)}$. In addition,
bounds on $I$ and $J_i$ guarantee that the cardinality of $Z_1$ is at
most polynomial in $||\uu||$.

Since $\lambda_1$ is not a root of unity, for each $\zeta \in Z_1$
there is at most one value of $n$ such that $\lambda_1^n =
\zeta$. Thm.~\ref{Ge} in App.~\ref{tools} then entails that
this value (if it exists) is at most $M = ||\uu||^{\mathcal{O}(1)}$,
which we can take to be uniform across all $\zeta \in Z_1$. We can
now invoke Cor.~\ref{Baker-cor} in App.~\ref{tools} to
conclude that, for $n > M$, and for all $\zeta \in Z_1$, we have
\begin{equation}
|\lambda_1^n - \zeta| > \frac{1}{n^{||\uu||^{D}}} \, ,
\end{equation}
where $D \in \mathbb{N}$ is some absolute constant.

Let $b > 0$ be minimal such that the set 
\[
\{z_1 \in \mathbb{C} : |z_1| = 1 
\mbox{ and, for all $\zeta \in Z_1$},\, |z_1 - \zeta| \geq
\frac{1}{b} \} 
\]
is non-empty. Thanks to our bounds on the cardinality of $Z_1$, we can
use the first-order theory of the reals, together with
Thm.~\ref{renegar} in App.~\ref{tools}, to conclude that $b$ is
algebraic and $||b|| = ||\uu||^{\mathcal{O}(1)}$.

Define the function $g: [b, \infty) \rightarrow \mathbb{R}$ as
  follows:
\[
g(x) = \min \{h(z_1,z_2,z_3) - \mu : (z_1,z_2,z_3) \in T 
\mbox{ and, for all $\zeta \in Z_1$},\, |z_1 - \zeta| \geq 
\smash{\frac{1}{x}} \} \, .
\]

It is clear that $g$ is continuous and $g(x) > 0$ for all $x \in
[b,\infty)$. Moreover, as before, $g$ can be rewritten as a function
  in the first-order theory of the reals over a bounded number of
  variables, with $||g|| = \uu^{\mathcal{O}(1)}$. It follows from
  Prop.~2.6.2 of~\cite{BCR98} (invoked with the function $1/g$)
  that there is a polynomial $P \in \mathbb{Z}[x]$ such that, for all
  $x \in [b,\infty)$,
\begin{equation}
g(x) \geq \frac{1}{P(x)} \, .
\end{equation} 
Moreover, an examination of the proof of~\cite[Prop.~2.6.2]{BCR98}
reveals that $P$ is obtained through a process which hinges on
quantifier elimination. Combining this with Thm.~\ref{renegar} in
App.~\ref{tools}, we are therefore able to conclude that $||P|| =
||\uu||^{\mathcal{O}(1)}$, a fact which relies among others on our
upper bounds for $||b||$.

By Eqs.~(\ref{eq1})--(\ref{eq3}), we can find $\varepsilon \in
(0,1)$ and $N = 2^{||\uu||^{\mathcal{O}(1)}}$ such that for all $n >
N$, we have $|r(n)| < (1-\varepsilon)^n$, and moreover $1/\varepsilon
= 2^{||\uu||^{\mathcal{O}(1)}}$. Moreover, by
Prop.~\ref{exp-bound} in App.~\ref{tools}, there is $N' =
2^{||\uu||^{\mathcal{O}(1)}}$ such that
\begin{equation}
\label{combo-poly}
\frac{1}{P(n^{||\uu||^D})} > (1-\varepsilon)^n
\end{equation}
for all $n \geq N'$.

Combining Eqs.~(\ref{eq-seven})--(\ref{combo-poly}), we get
\begin{align*}
\frac{u^n}{\rho^n} = & \ a + h(\lambda_1^n, \lambda_2^n, \lambda_3^n) +
r(n)\\
\geq & \ {-\mu} + h(\lambda_1^n, \lambda_2^n, \lambda_3^n) - 
   (1-\varepsilon)^n \\
\geq & \ g(n^{||\uu||^D}) - (1-\varepsilon)^n \\ 
\geq & \ \frac{1}{P(n^{||\uu||^D})} - (1-\varepsilon)^n \\ 
\geq & \ 0 \, ,
\end{align*}
provided $n > \max\{M,N,N'\}$, which establishes ultimate positivity
of $\uu$ and provides an exponential upper bound on the index of
possible violations of positivity, as required. We can then decide the
actual positivity of $\uu$ via a $\mathrm{coNP}^{\mathrm{PosSLP}}$
  procedure as detailed earlier.

This completes the proof of Thm.~\ref{pos-theorem}.

\newpage

\appendix

\section{Mathematical Tools}
\label{tools}

In this appendix we summarise the main technical tools used in this
paper.

For $p \in \mathbb{Z}[x_1, \ldots, x_m]$ a polynomial with integer
coefficients, let us denote by $||p||$ the bit length of its
representation as a list of coefficients encoded in binary. Note that
the degree of $p$ is at most $||p||$, and the height of $p$---i.e.,
the maximum of the absolute values of its coefficients---is at
most $2^{||p||}$.

We begin by recalling some basic facts about algebraic numbers and
their (efficient) manipulation. The main references
include~\cite{Coh93,BPR06,Ren92}.

A complex number $\alpha$ is \defemph{algebraic} if it is a root of a
single-variable polynomial with integer coefficients. The
\defemph{defining polynomial} of $\alpha$, denoted $p_{\alpha}$, is
the unique polynomial of least degree, and whose coefficients do not
have common factors, which vanishes at $\alpha$. The \defemph{degree}
and \defemph{height} of $\alpha$ are respectively those of
$p_{\alpha}$.

A standard representation\footnote{Note that this representation is
  not unique.} for algebraic numbers is to encode $\alpha$ as a tuple
comprising its defining polynomial together with rational
approximations of its real and imaginary parts of sufficient precision
to distinguish $\alpha$ from the other roots of $p_{\alpha}$. More
precisely, $\alpha$ can be represented by $(p_{\alpha},a, b,r) \in
\mathbb{Z}[x] \times \mathbb{Q}^3$ provided that $\alpha$ is the
unique root of $p_{\alpha}$ inside the circle in $\mathbb{C}$ of
radius $r$ centred at $a + bi $. A separation bound due to
Mignotte~\cite{Mig82} asserts that for roots $\alpha \neq \beta$ of a
polynomial $p \in \mathbb{Z}[x]$, we have
\begin{equation}
\label{root-sep-bound}
|\alpha - \beta| > \frac{\sqrt6}{d^{(d+1)/2} H^{d-1}} \, ,
\end{equation} 
where $d$ and $H$ are respectively the degree and height of $p$. Thus
if $r$ is required to be less than a quarter of the root-separation
bound, the representation is well-defined and allows for equality
checking. Given a polynomial $p \in \mathbb{Z}[x]$, it is well-known
how to compute standard representations of each of its roots in time
polynomial in $||p||$~\cite{Pan97,Coh93,BPR06}. Thus given $\alpha$ an
algebraic number for which we have (or wish to compute) a standard
representation, we write $||\alpha||$ to denote the bit length of this
representation. From now on, when referring to computations on
algebraic numbers, we always implicitly refer to their standard
representations. 

Note that Eq.~(\ref{root-sep-bound}) can be used more generally
to separate arbitrary algebraic numbers: indeed, two algebraic numbers
$\alpha$ and $\beta$ are always roots of the polynomial
$p_{\alpha}p_{\beta}$ of degree at most the sum of the degrees of
$\alpha$ and $\beta$, and of height at most the product of the heights
of $\alpha$ and $\beta$.

Given algebraic numbers $\alpha$ and $\beta$, one can compute
$\alpha+\beta$, $\alpha \beta$, $1/\alpha$ (for non-zero $\alpha$),
$\overline{\alpha}$, and $|\alpha|$, all of which are algebraic, in
time polynomial in $||\alpha|| + ||\beta||$. Likewise, it is
straightforward to check whether $\alpha = \beta$. Moreover, if
$\alpha \in \mathbb{R}$, deciding whether $\alpha > 0$ can be done in
time polynomial in $||\alpha||$. Efficient algorithms for all these
tasks can be found in~\cite{Coh93,BPR06}.

Remarkably, integer multiplicative relationships among a fixed number
of algebraic numbers can be elicited systematically in polynomial
time:

\begin{theorem}
\label{Ge}
Let $m$ be fixed, and let $\lambda_1, \ldots, \lambda_m$ be
complex algebraic numbers of modulus $1$. Consider the free abelian group
$L$ under addition given by
\[
L =  \{(v_1, \ldots, v_m) \in \mathbb{Z}^m : 
\lambda_1^{v_1} \ldots \lambda_m^{v_m} =
1\} \, .
\]

$L$ has a basis $\{\vec{\ell_1}, \ldots, \vec{\ell_p}\} \subseteq
\mathbb{Z}^m$ (with $p \leq m$), where the entries of each of the
$\vec{\ell_j}$ are all polynomially bounded in $||\lambda_1|| + \ldots
+ ||\lambda_m||$. Moreover, such a basis can be computed in time
polynomial in $||\lambda_1|| + \ldots +
||\lambda_m||$.
\end{theorem}

Note in the above that the bound is on the \emph{magnitude} of the
vectors $\vec{\ell_j}$ (rather than the bit length of their
representation), which follows from a deep result of
Masser~\cite{Mas88}. For a proof of Thm.~\ref{Ge}, see
also~\cite{Ge93,CLZ00}.

We now turn to the first-order theory of the reals. Let $\vec{x} =
(x_1, \ldots, x_m)$ and $\vec{y} = (y_1, \ldots, y_r)$ be tuples of
real-valued variables, and 
%let $g_1, \ldots, g_{\ell} \in \mathbb{Z}[\vec{x},\vec{y}]$ 
%be polynomials with integer coefficients
%over all these variables. Let 
let $\sigma(\vec{x}, \vec{y})$ be a Boolean combination of
atomic predicates of the form $g(\vec{x}, \vec{y}) \sim 0$, where each
$g(\vec{x}, \vec{y}) \in \mathbb{Z}[\vec{x},\vec{y}]$ is a polynomial
with integer coefficients over these variables, and $\sim$ is either
$>$ or $=$.  A \defemph{formula of the first-order theory of the
  reals} is of the form
\begin{equation}
\label{for-formula}
Q_1 \vec{x_1} \ldots Q_m \vec{x_m} \, \sigma (\vec{x}, \vec{y}) \, ,
\end{equation}
where each $Q_i$ is one of the quantifiers $\exists$ or $\forall$.
Let us denote the above formula by $\tau(\vec{y})$, whose free
variables are contained in $\vec{y}$. When $\tau$ has no free
variables, we refer to it as a \defemph{sentence}.  Naturally,
$||\tau(\vec{y})||$ denotes the bit length of the syntactic
representation of the formula, and the \defemph{degree} and
\defemph{height} of $\tau(\vec{y})$ refer to the maximum degree and
height of the polynomials $g(\vec{x}, \vec{y})$ appearing in
$\tau(\vec{y})$.

Tarski~\cite{Tar51} famously showed that the first-order theory of
the reals admits \emph{quantifier elimination}: that is, given
$\tau(\vec{y})$ as above, there is a quantifier-free formula
$\chi(\vec{y})$ that is equivalent to $\tau$: for any tuple
$\vec{\hat{y}} = (\hat{y}_1, \ldots, \hat{y}_r) \in \mathbb{R}^r$ of
real numbers, $\tau(\vec{\hat{y}})$ holds iff $\chi(\vec{\hat{y}})$
holds. An immediate corollary is that the first-order theory of the
reals is decidable.

Tarski's procedure, however, has non-elementary complexity. Many
substantial improvements followed over the years, starting with
Collins's technique of cylindrical algebraic
decomposition~\cite{Col75}. For our purposes, we require bounds not
only on the computation time, but also on the degree and height of the
resulting equivalent quantifier-free formula, as well as on the number
of atomic predicates it comprises. Such bounds are available thanks to
the work of Renegar~\cite{Ren92}. In this paper, we focus exclusively
on the situation in which the number of variables is uniformly
bounded.

\begin{theorem}[Renegar]
\label{renegar}
Let $M \in \mathbb{N}$ be fixed. Let $\tau(\vec{y})$ be of the
form~(\ref{for-formula}) above. Assume that the number of (free and
bound) variables in $\tau(\vec{y})$ is bounded by $M$ (i.e., $m + r
\leq M$). Denote the degree of $\tau(\vec{y})$ by $d$ and the number
of atomic predicates in $\tau(\vec{y})$ by $n$.

Then there is a procedure which computes an equivalent
quantifier-free formula 
\[ \chi(\vec{y}) = \bigvee_{i=1}^I \bigwedge_{j=1}^{J_i}
h_{i,j}(\vec{y}) \sim_{i,j} 0 \] 
in disjunctive normal form, where each $\sim_{i,j}$ is either $>$ or $=$,
with the following properties:
\begin{enumerate}
\item Each of $I$ and $J_i$ (for $1 \leq i \leq I$) is bounded
  by $(nd)^{\mathcal{O}(1)}$;

\item The degree of $\chi(\vec{y})$ is bounded by 
$(nd)^{\mathcal{O}(1)}$;

\item The height of $\chi(\vec{y})$ is bounded by
$2^{||\tau(\vec{y})||(nd)^{\mathcal{O}(1)}}$.

\end{enumerate}

Moreover, this procedure runs in time polynomial in
$||\tau(\vec{y})||$.
\end{theorem}

Note in particular that, when $\tau$ is a sentence, its truth value
can be determined in polynomial time.

Thm.~\ref{renegar} follows immediately from Thm.~1.1 (for the
case in which $\tau$ is a sentence) and Thm.~1.2 of~\cite{Ren92}.

Our next result is a special case of Kronecker's famous theorem on
simultaneous Diophantine approximation, a statement and proof of
which can be found in~\cite[Chap.~7, Sec.~1.3, Prop.~7]{Bou66}.

For $x \in \mathbb{R}$, write $[x]_{2 \pi}$ to denote the distance
from $x$ to the closest integer multiple of $2 \pi$: $[x]_{2 \pi} =
\min \{ |x - 2\pi j| : j \in \mathbb{Z} \}$.

\begin{theorem}[Kronecker]
\label{kronecker}
Let $t_1, \ldots, t_m, x_1, \ldots, x_m \in [0, 2\pi)$. The following
  are equivalent:
\begin{enumerate}
\item 
For any $\varepsilon > 0$, there exists $n \in \mathbb{Z}$ such
  that, for $1 \leq j \leq m$, we have $[n t_j - x_j]_{2\pi} 
\leq \varepsilon$.

\item For every tuple $(v_1, \ldots, v_m)$ of integers such that $v_1
  t_1 + \ldots + v_m t_m \in 2\pi \mathbb{Z}$, we have
$v_1 x_1 + \ldots + v_m x_m \in 2\pi \mathbb{Z}$.
\end{enumerate}
\end{theorem}

We can strengthen Thm.~\ref{kronecker} by requiring that $n \in
\mathbb{N}$ in the first assertion. Indeed, suppose that in a given
instance, we find that $n < 0$. A straightforward pigeonhole argument
shows that there exist arbitrarily large positive integers $g$ such
that $[g t_j]_{2\pi} \leq \varepsilon$ for $1 \leq j \leq m$. It
follows that $[(g+n) t_j - x_j]_{2 \pi} \leq 2 \varepsilon$, which
establishes the claim for sufficiently large $g$ (noting that
$\varepsilon$ is arbitrary).

Let $\lambda_1, \ldots, \lambda_m$ be complex algebraic numbers of
modulus $1$. For each $j \in \{1, \ldots, m\}$, write $\lambda_j =
e^{i\theta_j}$ for some $\theta_j \in [0, 2\pi)$.  Let 
\begin{align*}
L = & \ \{(v_1,
  \ldots, v_m) \in \mathbb{Z}^m : \lambda_1^{v_1} \ldots
  \lambda_m^{v_m} = 1\} \\
  = & \  \{(v_1, \ldots, v_m) \in
  \mathbb{Z}^m : v_1 \theta_1 + \ldots + v_m \theta_m \in 2\pi 
  \mathbb{Z} \} \, .
\end{align*}
Recall from Thm.~\ref{Ge} that $L$ is a free abelian group under
addition with basis $\{\vec{\ell_1}, \ldots, \vec{\ell_p}\} \subseteq
\mathbb{Z}^m$, where $p \leq m$.

For each $j \in \{1, \ldots, p\}$, let $\vec{\ell_j} = (\ell_{j,1},
\ldots, \ell_{j,m})$.  Write
\[
R = \{ \vec{x} = (x_1, \ldots, x_m) \in [0,2\pi)^m : 
\vec{\ell_j} \cdot \vec{x} \in 2\pi \mathbb{Z} \mbox{ for } 1 \leq j \leq
p\} \, .
\] 
By Thm.~\ref{kronecker}, for an arbitrary tuple $(x_1, \ldots,
x_m) \in [0,2\pi)^m$, it is the case that, for all $\varepsilon > 0$,
  there exists $n \in \mathbb{N}$ such that, for $j \in \{1, \ldots,
  m\}$, $[n \theta_j - x_j]_{2\pi} \leq \varepsilon$ iff
  $(x_1, \ldots, x_m) \in R$. 

Write $\mathbb{T} = \{ z \in \mathbb{C} : |z|=1\}$, and 
observe that $(x_1, \ldots, x_m) \in R$ iff $(e^{i
  x_1}, \ldots, e^{i x_m}) \in T$, where
\[
T = \{(z_1, \ldots, z_m) \in \mathbb{T}^m
: 
\mbox{for each $j \in \{1, \ldots, p\}$, }
z_1^{\ell_{j,1}} \ldots z_m^{\ell_{j,m}} = 1 \} \, .
\]

Since $e^{i n \theta_j} = \lambda_j^n$,
we immediately have the following:

\begin{corollary}
\label{density}
Let $\lambda_1, \ldots, \lambda_m$ and $T$ be as above. Then
$\{(\lambda_1^n, \ldots, \lambda_m^n) : n \in \mathbb{N}\}$ is
a dense subset of $T$.
\end{corollary}

Finally, we give a version of Baker's deep theorem on linear forms in
logarithms. The particular statement we have chosen is a sharp formulation
due to Baker and W\"ustholz~\cite{BW93}.

In what follows, $\log$ refers to the principal value of the complex
logarithm function given by $\log z = \log |z| + i \arg z$, where
$-\pi < \arg z \leq \pi$.

\begin{theorem}[Baker and W\"ustholz]
\label{Baker}
Let $\alpha_1, \ldots, \alpha_m \in \mathbb{C}$ be algebraic numbers
different from $0$ or $1$, and let $b_1,\ldots,b_m \in \mathbb{Z}$ be
integers. Write
\[\Lambda = b_1 \log \alpha_1 + \ldots + b_m \log \alpha_m \, . \]
Let $A_1, \ldots, A_m, B \geq e$ be real numbers such that, for each
$j \in \{1, \ldots, m\}$, $A_j$ is an upper bound for the height of
$\alpha_j$, and $B$ is an upper bound for $|b_j|$. Let $d$ be the
degree of the extension field $\mathbb{Q}(\alpha_1, \ldots, \alpha_m)$
over $\mathbb{Q}$.

If $\Lambda \neq 0$, then 
$\log |\Lambda| > -(16md)^{2(m+2)}\log A_1
\ldots \log A_m \log B$.  
\end{theorem}

%Note that $d$ above is bounded by the product of the degrees of the
%$\alpha_j$. 

\begin{corollary}
\label{Baker-cor}
There exists $D \in \mathbb{N}$ such that, for any algebraic numbers
$\lambda, \zeta \in \mathbb{C}$ of modulus
$1$, and for all $n \geq 2$, whenever $\lambda^n \neq \zeta$, then
\[ |\lambda^n - \zeta| > \frac{1}{n^{(||\lambda|| + ||\zeta||)^D}}
\, .
\]
\end{corollary}

\begin{proof}
We can clearly assume that $\lambda \neq 1$, otherwise the result
follows immediately from Eq.~(\ref{root-sep-bound}). Likewise,
the case $\zeta = 1$ is easily handled along the same lines as the
proof below, so we assume $\zeta \neq 1$.

Let $\theta = \arg \lambda$ and $\varphi = \arg \zeta$. Then for all
$n \in \mathbb{N}$, there is $j \in \mathbb{Z}$ with $|j| \leq n$ such that
\[ |\lambda^n - \zeta| > \frac{1}{2} |n \theta - \varphi - 2j\pi| 
= \frac{1}{2} |n \log \lambda  - \log \zeta - 2j \log(-1)| \, . \]

Let $H \geq e$ be an upper bound for the heights of $\lambda$ and
$\zeta$, and let $d$ be the largest of the degrees of $\lambda$ and
$\zeta$. Notice that the degree of $\mathbb{Q}(\lambda,\zeta)$ over
$\mathbb{Q}$ is at most $d^2$. Applying Thm.~\ref{Baker} to the
right-hand side of the above equation, we get
\[ |\lambda^n - \zeta| > \frac{1}{2}\exp\left(
-(48d^2)^{10} \log^2 H \log (2n+1) \right) =
\frac{1}{2(2n+1)^{(\log^2 H)(48d^2)^{10}}} \, . 
\]
for $n \geq 1$, provided $\lambda^n \neq \zeta$.

The required result follows by noting that $\log H \leq ||\lambda|| +
||\zeta||$ and $d \leq ||\lambda|| + ||\zeta||$.
\qed
\end{proof}

Finally, we record the following fact, whose straightforward proof is
left to the reader.

\begin{proposition}
\label{exp-bound}
Let $a \geq 2$ and $\varepsilon \in (0,1)$ be real numbers. Let $B \in
\mathbb{Z}[x]$ have degree at most $a^{D_1}$ and height at most
$2^{a^{D_2}}$, and assume that $1/\varepsilon \leq 2^{a^{D_3}}$, for
    some $D_1,D_2,D_3 \in \mathbb{N}$. Then there is $D_4 \in
    \mathbb{N}$ depending only on $D_1, D_2, D_3$ such that, for all
    $n \geq 2^{a^{D_4}}$, $\displaystyle{\frac{1}{B(n)} >
 (1-\varepsilon)^n}$.
\end{proposition}

\newpage

\section{Zero-Dimensionality Lemmas}
\label{app-lemmas}

We present the key results enabling our application of Baker's theorem
to the discrete orbit $\{(\lambda_1^n,\ldots,\lambda_m^n) : n \in
\mathbb{N}\}$ in two or three complex dimensions. In the terminology
of Sec.~\ref{decidability}, Lem.~\ref{zero-dim-lemma} shows that
the function $h$ achieves its minimum over the torus $T$ at finitely
many points. To do so, it relies on Lem.~\ref{lem-one-constraint} to
handle the case in which $L$, the free abelian group of multiplicative
relationships among $\lambda_1, \ldots, \lambda_m$, has rank $1$ or
$0$, and invokes Lem.~\ref{lem-m-1-constraints} to do the same when
$L$ has rank $m-1$.

\begin{lemma}
\label{lem-one-constraint}
Let $a_1, \ldots, a_m \in \mathbb{R}$ and $\varphi_1, \ldots,
\varphi_m \in \mathbb{R}$ be two collections of real numbers, with
each of the $a_i$ non-zero, and let $\ell_1, \ldots, \ell_m \in
\mathbb{Z}$ be $m$ integers. Define $f,g : \mathbb{R}^m \rightarrow
\mathbb{R}$ by setting
\[
f(x_1, \ldots, x_m) = \sum_{i=1}^m a_i \cos(x_i + \varphi_i) \quad
\mbox{ and } \quad
g(x_1, \ldots, x_m) = \sum_{i=1}^m \ell_i x_i \, .
\]
Assume that $g(x_1, \ldots, x_m)$ is not of the form $\ell(x_i \pm
x_j)$, for some non-zero $\ell \in \mathbb{Z}$ and indices $i \neq
j$. Let $\psi \in \mathbb{R}$, and let $\mu \in \mathbb{R}$ be the
minimum achieved by the function $f$ subject to the constraint $g(x_1,
\ldots,x_m) = \psi$.

Then $f$, subject to $g(x_1, \ldots,x_m) = \psi$, achieves $\mu$ at
only finitely many points over the domain $[0,2\pi)^m$.
\end{lemma}

\begin{proof}
We will establish the slightly stronger statement that $f$, subject to
the constraint $g = \psi$, achieves its minimum over $\mathbb{R}^m$
only finitely often modulo $2\pi$.

Note that by performing the substitutions $x'_i = x_i + \varphi_i$
(for $1 \leq i \leq m$) and $\psi' = \psi + \sum_{i=1}^m \ell_i \varphi_i$,
and rephrasing the statement in terms of the primed variables and
constant $\psi'$, we see that we may assume without loss of generality
that each $\varphi_i = 0$.

Observe that if each $\ell_i = 0$ (corresponding to there being no
constraint), the result is immediate: $f$ is minimised when each $x_i$
is either an odd or even multiple of $\pi$, depending on the sign of
$a_i$. Without loss of generality, let us therefore assume that
$\ell_1$ is non-zero. The case of $m = 1$ is also immediate, since the
constraint then reduces the domain of the unique variable $x_1$ to a
singleton. Let us therefore assume that $m \geq 2$.

We use the method of Lagrange multipliers. Minima of $f$ subject to
the constraint $g = \psi$ must satisfy $\nabla f = \lambda \nabla g$
for some $\lambda \in \mathbb{R}$, i.e., $-a_i\sin x_i = \lambda
\ell_i$, for $1 \leq i \leq m$. Note that $\lambda$ must satisfy
$|\lambda| \leq \frac{|a_i|}{|\ell_i|}$ for all $1 \leq i \leq
m$. Observe also that each choice of $\lambda$ gives rise to only
finitely many choices of $x_1, \ldots, x_m$ (modulo $2\pi$) which
satisfy these equations.

From $-a_i\sin x_i = \lambda \ell_i$, it follows that 
$\cos^2 x_i = 1 - \frac{\lambda^2 \ell_i^2}{a_i^2}$. Taking
square roots gives us $2^m$ choices of signs, and for each choice
let us write
\[\tilde{f}(\lambda) = \sum_{i=1}^m \pm a_i 
\sqrt{1 - \lambda^2 \frac{\ell_i^2}{a_i^2}}\, .
\]

%Let us denote by $\mu$ the global minimum of $f$ subject to $g =
%\psi$.
Suppose that there are infinitely many values of $(x_1,
\dots, x_m)$ (modulo $2 \pi$) such that $g(x_1, \ldots,x_m) = \psi$
and $f(x_1, \ldots, x_m) = \mu$. It then follows that, for some fixed
choice of signs, there must be infinitely many values of $\lambda$
such that $\tilde{f}(\lambda) = \mu$.

Assume without loss of generality that $\frac{|a_1|}{|\ell_1|} \leq
\frac{|a_i|}{|\ell_i|}$ for $2 \leq i \leq m$. Notice that
$\tilde{f}(\lambda)$ is analytic (equal to its Taylor power series) on
$(\frac{-|a_1|}{|\ell_1|},\frac{|a_1|}{|\ell_1|})$.  Now if the set of
$\lambda$ such that $\tilde{f}(\lambda) = \mu$ has an accumulation
point in $(\frac{-|a_1|}{|\ell_1|},\frac{|a_1|}{|\ell_1|})$, then
$\tilde{f}$ is identically equal to $\mu$ on
$[\frac{-|a_1|}{|\ell_1|},\frac{|a_1|}{|\ell_1|}]$. Thus in any case
the set of $\lambda$ such that $\tilde{f}(\lambda) = \mu$ must have an
accumulation point at $\frac{|a_1|}{|\ell_1|}$.

Observe that if $\frac{|a_1|}{|\ell_1|} < \frac{|a_i|}{|\ell_i|}$ for
all $2 \leq i \leq m$, then a contradiction is reached as $\tilde{f}$
cannot infinitely often take on the constant value $\mu$ as $\lambda$
approaches $\frac{|a_1|}{|\ell_1|}$. To see this, examine the
derivative of each term of the form $\sqrt{1 -
  \lambda^2\frac{\ell_i^2}{a_i^2}}$: it remains bounded for $i \neq
1$, but tends to $- \infty$ for $i = 1$.

Let $I$ be the set of indices $i \in \{1, \ldots, m\}$ such that
$\frac{|a_i|}{|\ell_i|} = \frac{|a_1|}{|\ell_1|}$. By the same
argument as above, for the given choice of signs in $\tilde{f}$, we
must have $\displaystyle{\sum_{i \in I}} \pm a_i = 0$, and therefore
for all $\lambda \in
[\frac{-|a_1|}{|\ell_1|},\frac{|a_1|}{|\ell_1|}]$,
\begin{equation}
\label{f-tilde}
\tilde{f}(\lambda) = \sum_{i\notin I} \pm a_i 
\sqrt{1 - \lambda^2 \frac{\ell_i^2}{a_i^2}}\, .
\end{equation}

Observe that $|I| \geq 2$. Two cases now arise, according to whether 
(i)~$|I| \geq 3$ or (ii)~$|I| = 2$. In both cases, we derive a
contradiction by showing that $f$ subject to $g=\psi$ can achieve a
value strictly lower than $\mu$.

(i)~Suppose without loss of generality that $I = \{1, 2, \ldots, p\}$,
where $p \geq 3$, and that $|a_p| \leq |a_i|$ for $1 \leq i \leq p-1$.
Pick $\hat{x}_1, \ldots, \hat{x}_m \in
\mathbb{R}$ such that $f(\hat{x}_1, \ldots, \hat{x}_m) = \mu$ and
$g(\hat{x}_1, \ldots, \hat{x}_m) = \psi$. There is some value
$\hat{\lambda} \in [\frac{-|a_1|}{|\ell_1|},\frac{|a_1|}{|\ell_1|}]$
such that $-a_i \sin \hat{x}_i = \hat{\lambda} \ell_i$, for $1 \leq i
\leq m$. Now for the given choice of signs in $\tilde{f}$,
\[ 
\sum_{i=1}^p  \pm a_i \sqrt{1 - {\hat{\lambda}}^2 \frac{\ell_i^2}{a_i^2}} = 0
\quad \mbox{ and } \quad 
\sum_{i=p+1}^m  \pm a_i \sqrt{1 - {\hat{\lambda}}^2 \frac{\ell_i^2}{a_i^2}} =
\mu \, ,
\]
or equivalently,
\begin{equation}
\label{mu-eqn}
\sum_{i=1}^p  a_i \cos \hat{x}_i = 0
\quad \mbox{ and } \quad 
\sum_{i=p+1}^m  a_i \cos \hat{x}_i = \mu \, .
\end{equation}

In order to make $f$ assume a value strictly smaller than $\mu$, 
pick $\check{x}_1, \ldots, \check{x}_{p-1}$ to be $\pi$ or $0$
depending respectively on the signs of $a_1, \ldots, a_{p-1}$, and
pick $\check{x}_p$ so that 
$g(\check{x}_1, \ldots, \check{x}_p, \hat{x}_{p+1}, \ldots,
\hat{x}_m) = \psi$ (noting that $\ell_p \neq 0$ since $p \in I$). Then 
\[ 
\sum_{i=1}^p a_i \cos \check{x}_i \leq 
 - \left(\sum_{i=1}^{p-1}|a_i| \right) + |a_p| < 0 \, ,
\]
where the strict inequality follows from the fact that 
$p \geq 3$ and $|a_p| \leq |a_i|$ for $1 \leq i \leq p-1$.

It then follows by the right-hand side of~(\ref{mu-eqn}) that
\[
f(\check{x}_1, \ldots, \check{x}_p, \hat{x}_{p+1}, \ldots, \hat{x}_m)
< \mu \, ,
\]
concluding Case~(i).

(ii)~Without loss of generality, let us have $I = \{1, 2\}$, so that
$|a_1| = |a_2|$ and $|\ell_1| = |\ell_2|$. Note that we then cannot
have $\ell_3, \ldots, \ell_m$ all zero, otherwise $g$ would be of the
form $\ell_1(x_1 \pm x_2)$, violating one of our hypotheses. It
therefore also follows that $m \geq 3$.

We can thus assume without loss of generality that $\ell_3$ is
non-zero, and furthermore that $\frac{|a_3|}{|\ell_3|} \leq
\frac{|a_i|}{|\ell_i|}$ for $4 \leq i \leq m$. From
Eq.~(\ref{f-tilde}), we see that $\tilde{f}$ can be analytically
extended to the larger domain $(\frac{-|a_3|}{|\ell_3|},
\frac{|a_3|}{|\ell_3|})$, and by a similar line of reasoning as
earlier, we can then conclude that there must be a non-empty set $J
\subseteq \{3, \ldots, m\}$ such that, for all $j \in J$,
$\frac{|a_j|}{|\ell_j|} = \frac{|a_3|}{|\ell_3|}$ and moreover
$\displaystyle{\sum_{j \in J}} \pm a_j = 0$ for the given choice of
signs in $\tilde{f}$. We can therefore write
\[
\tilde{f}(\lambda) = \sum_{i\notin I \cup J} \pm a_i 
\sqrt{1 - \lambda^2 \frac{\ell_i^2}{a_i^2}}\, .
\]

But this situation is entirely similar to Case~(i) since $|I \cup J|
\geq 3$, which concludes Case~(ii) and the proof of
Lem.~\ref{lem-one-constraint}.
\qed
\end{proof}

\begin{lemma}
\label{lem-m-1-constraints}
Let $\uu$ be a non-degenerate simple LRS, with dominant characteristic
roots $\rho \in \mathbb{R}$ and $\gamma_1, \overline{\gamma_1},
\ldots, \gamma_m, \overline{\gamma_m} \in \mathbb{C} \setminus
\mathbb{R}$. Write $\lambda_i = \gamma_i/\rho$ for $1 \leq i \leq m$,
and assume the free abelian group $L = \{(v_1, \ldots, v_m) \in
\mathbb{Z}^m : \lambda_1^{v_1} \ldots \lambda_m^{v_m} = 1\}$ has rank
$m-1$.  Let $\{\vec{\ell_1}, \ldots, \vec{\ell_{m-1}}\}$ be a basis
for $L$, and write $\vec{\ell_j} = (\ell_{j,1}, \ldots, \ell_{j,m})$
for $1 \leq j \leq m-1$. Let
\[
 M = 
\begin{pmatrix}
\ell_{1,1} & \ell_{1,2} & \ldots & \ell_{1,m-1} & \ell_{1,m} \\
\ell_{2,1} & \ell_{2,2} & \ldots & \ell_{2,m-1} & \ell_{2,m} \\
\vdots & \vdots & \ddots & \vdots & \vdots \\
\ell_{m-1,1} & \ell_{m-1,2} & \ldots & \ell_{m-1,m-1} & \ell_{m-1,m}  \\
\end{pmatrix} \, .
\]
Let $a_1, \ldots, a_m \in \mathbb{R}$ and $\varphi_1, \ldots,
\varphi_m \in \mathbb{R}$ be two collections of $m$ real numbers,
with each of the $a_i$ non-zero, and define
$f : \mathbb{R}^m \rightarrow \mathbb{R}$, by setting
\[ f(x_1, \ldots, x_m) = \sum_{i=1}^m a_i \cos(x_i + \varphi_i) \,
. \] Let $\vec{q} = (q_1, \ldots, q_{m-1}) \in \mathbb{Z}^{m-1}$ be a
column vector of $m-1$ integers, and denote by $\vec{x}$ the column
vector of variables $(x_1, \ldots, x_m)$. Let $\mu \in \mathbb{R}$ be
the minimum achieved by the function $f$ subject to the constraint $M
\vec{x} = 2 \pi \vec{q}$.

Then $f$, subject to $M \vec{x} = 2 \pi \vec{q}$, achieves $\mu$
at only finitely many points over the domain $[0,2\pi)^m$.
\end{lemma}

\begin{proof}
%We shall establish the slightly stronger result that, for any real
%number $\mu$, the function $f(x_1, \ldots, x_m)$, subject to the
%constraint $M \vec{x} = 2 \pi \vec{q}$, takes on the value $\mu$ at
%most finitely often over the domain $[0,2\pi)^m$.
%
By repeatedly making use of the following row operations:
\begin{enumerate}
\item Swapping two rows,
\item Multiplying any row by a non-zero integer, and
\item Adding to any row any integer linear combination of any of the
  other rows,
\end{enumerate}
we can transform the augmented matrix $(M | \vec{q})$ into an integer matrix
\[
(N | \vec{p}) = \left(
\begin{matrix}
n_{1,1} & 0        & \hdotsfor{2}    & 0      & b_1 \\
0       & n_{2,2}   &    0   & \hdotsfor{1}    & 0      & b_2 \\
\vdots  & \ddots   & \ddots & \ddots    & \vdots & \vdots \\
0       & \ldots   &    0   & n_{m-2,m-2} & 0      & b_{m-2} \\
0       & \hdotsfor{2}    &   0    & n_{m-1,m-1} & b_{m-1}  
\end{matrix}\, \right| 
\left.
\begin{matrix}
p_1 \\
p_2 \\
\vdots \\
p_{m-2} \\
p_{m-1}
\end{matrix} \right) \, .
\]

Without loss of generality (relabelling variables and constants if
necessary), we can assume that this was achieved without the need for
any row-swapping operations.

Note that the rows of $N$ remain in $L$ (though need no longer form a
basis).  Hence for each $i \in \{1, \ldots, m-1\}$, the $\lambda_1,
\ldots, \lambda_m$ satisfy the equation $\lambda_i^{n_{i,i}}
\lambda_m^{b_i} = 1$. Since $M$ has rank $m-1$, and $N$ is obtained
from $M$ by elementary row operations, no row of $N$ can be $\vec{0}$.
From this and the fact that the LRS $\uu$ is non-degenerate we may
conclude that no $n_{i,i}$ can be zero (otherwise $\lambda_m$ would be
a root of unity), and likewise no $b_i$ can be zero (otherwise
$\lambda_i$ would be a root of unity). Furthermore, we can never have
$n_{i,i} = -b_i$ (otherwise $\lambda_i/\lambda_m$ would be a root of
unity) nor can we have $n_{i,i} = b_i$ (otherwise
$\overline{\lambda_i}/\lambda_m$ would be a root of unity). In other
words, we always have $n_{i,i}^2 \neq b_i^2$. Finally, for $i \neq j$,
$\frac{b_i}{n_{i,i}} \neq \frac{b_j}{n_{j,j}}$: indeed, since
$\lambda_i^{n_{i,i}} \lambda_m^{b_i} = 1$, we have
$\lambda_i^{n_{i,i}b_j} \lambda_m^{b_ib_j} = 1$, and likewise
$\lambda_j^{n_{j,j}b_i} \lambda_m^{b_ib_j} = 1$, from which we deduce
that $\lambda_i^{n_{i,i}b_j} = \lambda_j^{n_{j,j}b_i}$. But if we had
$\frac{b_i}{n_{i,i}} = \frac{b_j}{n_{j,j}}$, it would follow that
$\lambda_i/\lambda_j$ is a root of unity. Similarly, by noting that
$\overline{\lambda_j}^{n_{j,j}} = \lambda_j^{-n_{j,j}}$ and repeating
the calculation, we deduce that $\frac{b_i}{n_{i,i}} \neq
-\frac{b_j}{n_{j,j}}$ for $i \neq j$. Combining the last two
disequalities, we have that $\frac{b_i^2}{n_{i,i}^2} \neq
\frac{b_j^2}{n_{j,j}^2}$ for $i \neq j$.

It is clear that the equations $M \vec{x} = 2 \pi \vec{q}$ and $N
\vec{x} = 2 \pi \vec{p}$ are equivalent (as constraints over the vector
of real-valued variables $\vec{x}$). From the latter, we may
write $x_i = \frac{p_i}{n_{i,i}} - \frac{b_i}{n_{i,i}} x_m$ for $1
\leq i \leq m-1$. For ease of notation, let us set
\begin{gather*}
d_i = -\frac{b_i}{n_{i,i}} \ \mbox{for $1 \leq i \leq m-1$, and }
d_m = 1 \, ;\\
\nu_i =  \frac{p_i}{n_{i,i}} + \varphi_i \ \mbox{for $1 \leq i \leq m-1$, and }
\nu_m = \varphi_m \, .
\end{gather*}
From our earlier observations, let us record that:
\begin{enumerate}
\item Each $d_i$ is non-zero, and 
\item For $1 \leq i < j \leq m$, we have $d_i^2 \neq
d_j^2$. 
\end{enumerate}
Indeed, we have already seen that the second assertion holds when $j
\leq m-1$. But since $n_{i,i}^2 \neq b_i^2$, for $1 \leq i \leq m-1$
we have that $d_i^2 \neq 1 = d_m^2$.

Substituting into $f$ yields
\[ \tilde{f}(x_m) = 
\sum_{i=1}^m a_i 
\cos( d_i x_m + \nu_i) \, ,
\] 
where $\tilde{f}$ is now unconstrained. Since any value of $x_m$ in
$[0,2\pi)$ such that $\tilde{f}(x_m) = \mu$ yields at most one point
  $\vec{x}$ in $[0, 2\pi)^m$ satisfying $M \vec{x} = 2 \pi \vec{q}$
    and such that $f(\vec{x}) = \mu$, it remains to show that
    $\tilde{f}$ can achieve $\mu$ only finitely often over $[0,
      2\pi)$.

Thus suppose, to the contrary, that $\tilde{f}$ achieves $\mu$ at
infinitely many points in $[0, 2 \pi)$. These points must accumulate,
  and since $\tilde{f}$ is analytic over $\mathbb{R}$, $\tilde{f}$
  must be identically equal to $\mu$ all over the reals. It follows
  that derivatives of all orders must vanish everywhere. Now for $j
  \geq 1$, the $(2j-1)$th derivative of $\tilde{f}$ is given by
\[f^{(2j-1)}(x_m) = \sum_{i=1}^m (-1)^j d_i^{2j-1}a_i \sin(d_i x_m +
\nu_i) \, .
\]
Writing 
\[
D = 
\begin{pmatrix}
1 & \ & 1 & \hdotsfor{3} & 1 \\
-d_1^2 & \ & -d_2^2 & \hdotsfor{3} & -d_m^2 \\
d_1^4 & \ & d_2^4 & \hdotsfor{3} & d_m^4 \\
\vdots & \ & \vdots & \ & \vdots & \ & \vdots \\
(-1)^{m-1}d_1^{2(m-1)} & \ \ \ & (-1)^{m-1}d_2^{2(m-1)} & \ \ \ & 
\hdots & \ \ \ & (-1)^{m-1}d_m^{2(m-1)} \\ 
\end{pmatrix} \, ,
\]
we therefore have that
\[
\begin{pmatrix}
f^{(1)}(x_m) \\
f^{(3)}(x_m) \\
\vdots  \\
f^{(2m-1)}(x_m)
\end{pmatrix}
=
D
\begin{pmatrix}
-d_1 a_1 \sin(d_1 x_m + \nu_1) \\
-d_2 a_2 \sin(d_2 x_m + \nu_2) \\
\vdots  \\
-d_m a_m \sin(d_m x_m + \nu_m) 
\end{pmatrix}
=
\begin{pmatrix}
0 \\
0 \\
\vdots  \\
0
\end{pmatrix} 
\]
must hold for all $x_m \in \mathbb{R}$.

But this is a contradiction since $D$ is a Vandermonde matrix which is
invertible (given that for $i \neq j$, we have $-d_i^2 \neq -d_j^2$) and
the vector
\[
\begin{pmatrix}
-d_1 a_1 \sin(d_1 x_m + \nu_1) \\
-d_2 a_2 \sin(d_2 x_m + \nu_2) \\
\vdots  \\
-d_m a_m \sin(d_m x_m + \nu_m) 
\end{pmatrix}
\]
clearly cannot be identically $\vec{0}$.
\qed
\end{proof}

\begin{lemma}
\label{zero-dim-lemma}
Following the notation of Sec.~\ref{decidability},
let $\uu$ be a non-degenerate simple LRS with a real positive dominant
characteristic root $\rho > 0$ and complex dominant roots $\gamma_1,
\overline{\gamma_1}, \gamma_2, \overline{\gamma_2}, \gamma_3,
\overline{\gamma_3} \in \mathbb{C} \setminus \mathbb{R}$. Write
$\lambda_j = \gamma_j/\rho$ for $1 \leq j \leq 3$.

Let $L = \{(v_1, v_2, v_3) \in \mathbb{Z}^3 : \lambda_1^{v_1}
\lambda_2^{v_2} \lambda_3^{v_3} = 1\}$ have rank $p$ (as a free abelian group),
and let $\{\vec{\ell_1}, \ldots, \vec{\ell_p}\}$ be a basis for $L$.
Write $\vec{\ell_q} = (\ell_{q,1}, \ell_{q,2}, \ell_{q,3})$ for $1 \leq q
\leq p$. 

Let $T = \{(z_1, z_2, z_3) \in \mathbb{T}^3
: \mbox{for each $q \in \{1, \ldots, p\}$, }
z_1^{\ell_{q,1}} z_2^{\ell_{q,2}} z_3^{\ell_{q,3}} = 1 \}$, where 
$\mathbb{T} = \{ z \in \mathbb{C} : |z|=1\}$.

Define $h: T \rightarrow \mathbb{R}$ by setting
$h(z_1, z_2, z_3) = \sum_{j=1}^3 (c_j z_j + \overline{c_j} \overline{z_j})$.

Then $h$ achieves its minimum $\mu$ at only finitely many points over $T$.
\end{lemma}

\begin{proof}

(i)~We first consider the case in which $L$ has rank $1$, and
handle the case of rank $0$ immediately afterwards. Let $\vec{\ell_1}
= (\ell_{1,1}, \ell_{1,2}, \ell_{1,3}) \in \mathbb{Z}^3$ span $L$.
%For $x \in
%\mathbb{R}$, let $[x]_{2 \pi}$ denote the distance from $x$ to the
%closest integer multiple of $2 \pi$: $[x]_{2 \pi} = \min \{ |x - 2\pi
%j| : j \in \mathbb{Z} \}$. 
Write
\[
R = \{ (x_1, x_2, x_3) \in [0,2\pi)^3 : 
\ell_{1,1} x_1 + \ell_{1,2} x_2 + \ell_{1,3} x_3 \in 2\pi \mathbb{Z} \} \, .
\] 
Clearly, for any $(x_1, x_2, x_3) \in [0,2 \pi)^3$, we have
$(x_1, x_2, x_3) \in R$ iff 
$(e^{i x_1}, e^{i x_2}, e^{i x_3}) \in T$. Define $f : R
  \rightarrow \mathbb{R}$ by setting
\[ f(x_1, x_2, x_3) = \sum_{j = 1}^3 2|c_j|\cos(x_j + \varphi_j) \,
. 
\] 
Plainly, for all $(x_1, x_2, x_3) \in R$, we have $f(x_1,
  x_2, x_3) = h(e^{i x_1}, e^{i x_2}, e^{i x_3})$, and therefore the
  minima of $f$ over $R$ are in one-to-one correspondence with those
  of $h$ over $T$.

Define $g : \mathbb{R}^3 \rightarrow \mathbb{R}$ by setting
\[
g(x_1, x_2, x_3) = \ell_{1,1} x_1 + \ell_{1,2} x_2 + \ell_{1,3} x_3 \, .
\]
Note that $g(x_1,x_2,x_3)$ cannot be of the form $\ell(x_i-x_j)$, for
non-zero $\ell \in \mathbb{Z}$ and $i \neq j$, otherwise (by
definition of $\vec{\ell_1}$) $\lambda_i^{\ell} \lambda_j^{-\ell} = 1$,
i.e., $\lambda_i/\lambda_j$ would be a root of unity, contradicting
the non-degeneracy of $\uu$. Likewise, $g$ cannot be of the form
$\ell(x_i+x_j)$, otherwise $\lambda_i/\overline{\lambda_j}$ would be a
root of unity.

Finally, observe that for $(x_1,x_2,x_3) \in [0,2\pi)^3$, we have
  $(x_1,x_2,x_3) \in R$ iff $\ell_{1,1} x_1 + \ell_{1,2} x_2 +
  \ell_{1,3} x_3 = 2\pi q$, for some $q \in \mathbb{Z}$ with $|q| \leq
  |\ell_{1,1}| + |\ell_{1,2}| + |\ell_{1,3}|$. For each of these
  finitely many $q$, we can invoke Lem.~\ref{lem-one-constraint} with
  $f$, $g$, and $\psi = 2\pi q$, to conclude that $f$ achieves its
  minimum $\mu$ at finitely many points over $R$, and therefore that
  $h$ achieves the same minimum at finitely many points over $T$.

The case of $L$ having rank $0$, i.e., when there are no non-trivial integer
multiplicative relationships among $\lambda_1, \lambda_2, \lambda_3$,
is now a special case of the above, where we have $\ell_{1,1} =
\ell_{1,2} = \ell_{1,3} = 0$.

(ii)~We now turn to the case of $L$ having rank $2$. We have
$\vec{\ell_1} = (\ell_{1,1}, \ell_{1,2}, \ell_{1,3}) \in \mathbb{Z}^3$
and $\vec{\ell_2} = (\ell_{2,1}, \ell_{2,2}, \ell_{2,3}) \in
\mathbb{Z}^3$ spanning $L$. Let $\vec{x}$ denote the column vector
$(x_1,x_2,x_3)$, and write
\[ R = \{ (x_1, x_2, x_3) \in [0,2\pi)^3 : 
          \vec{\ell_1} \cdot \vec{x} \in 2\pi \mathbb{Z}\mbox{ and }
          \vec{\ell_2} \cdot \vec{x} \in 2\pi \mathbb{Z} \}\, .
\]
Define $f : R \rightarrow \mathbb{R}$ by setting
$\displaystyle{ f(x_1, x_2, x_3) = \sum_{j = 1}^3 2|c_j|\cos(x_j +
  \varphi_j)}$.  As before, the minima of $f$ over $R$ are in
one-to-one correspondence with those of $h$ over $T$.

For $(x_1,x_2,x_3) \in [0,2\pi)^3$, we have $\vec{\ell_1} \cdot
  \vec{x} \in 2\pi \mathbb{Z}$ and $\vec{\ell_2} \cdot \vec{x} \in
  2\pi \mathbb{Z}$ iff there exist $q_1, q_2 \in \mathbb{Z}$, with
  $|q_1| \leq |\ell_{1,1}| + |\ell_{1,2}| + |\ell_{1,3}|$ and $|q_2|
  \leq |\ell_{2,1}| + |\ell_{2,2}| + |\ell_{2,3}|$, such that
  $\vec{\ell_1} \cdot \vec{x} = 2\pi q_1$ and $\vec{\ell_2} \cdot
  \vec{x} = 2\pi q_2$. For each of these finitely many $\vec{q} =
  (q_1,q_2)$, we can invoke Lem.~\ref{lem-m-1-constraints} with $f$,
  $M =
\begin{pmatrix}
\ell_{1,1} & \ell_{1,2} & \ell_{1,3} \\
\ell_{2,1} & \ell_{2,2} & \ell_{2,3} 
\end{pmatrix}$, and $\vec{q}$, to conclude that $f$ achieves 
its minimum $\mu$ at finitely many points over $R$, and therefore that $h$
achieves the same minimum at finitely many points over $T$.

(iii)~Finally, observe that $L$ cannot have rank $3$, since this would
immediately entail that every $\lambda_j$ is a root of unity
(contradicting the non-degeneracy of $\uu$), following a row-reduction
procedure similar to the one presented in the first stages of the
proof of Lem.~\ref{lem-m-1-constraints}.
\qed
\end{proof}

\newpage

\bibliography{posbib2.bib}  

\end{document}